\documentclass[a4paper]{article}

\usepackage[utf8]{inputenc}
\usepackage[T1]{fontenc}   
\usepackage[english]{babel} 
\usepackage{amsmath}
\usepackage{amsfonts}
\usepackage{amssymb}
\usepackage{amsthm}
\usepackage{euscript}
\usepackage{tikz}
\usepackage{url}
\usepackage{a4wide}
\usepackage{enumerate}
\usepackage{graphicx}
\usepackage{mathdots}

\theoremstyle{plain}
\newtheorem{thm}{Theorem}
\newtheorem{prop}[thm]{Proposition}
\newtheorem{cor}[thm]{Corollary}
\newtheorem{lem}[thm]{Lemma}

\theoremstyle{definition}
\newtheorem{defn}[thm]{Definition}
\newtheorem{ex}{Example}

\theoremstyle{remark}
\newtheorem{rem}{Remark}
\newtheorem{nota}{Notation}

\newcommand{\G}{\mathfrak{G}}
\renewcommand{\H}{\mathfrak{H}}

\newcommand{\Z}{\mathbb{Z}}

\newcommand{\F}{\mathbb{F}}

\newcommand{\B}{\mathbb{B}}
\newcommand{\Sp}{\mathbb{S}}
\newcommand{\Rk}{\textrm{Rk}}

\newcommand{\coker}{\textrm{Coker}\ }
\newcommand{\M}{\mathfrak{M}}
\newcommand{\C}{\mathfrak{C}}
\renewcommand{\d}{\mathfrak{d}}
\renewcommand{\b}{\mathfrak{b}}
\renewcommand{\c}{\mathfrak{c}}
\renewcommand{\a}{\mathfrak{a}}
\newcommand{\s}{\mathfrak{s}}
\newcommand{\x}{\mathfrak{x}}

\DeclareMathOperator{\Ker}{Ker}
\DeclareMathOperator{\im}{Im}

\title{A construction of Quantum LDPC codes from Cayley graphs}

\author{Alain Couvreur
\thanks{INRIA Saclay \^Ile-de-France \& Laboratoire CNRS LIX --- \'Ecole Polytechnique, Route de Saclay --- 91128 Palaiseau Cedex, France. {\tt alain.couvreur@inria.fr}} 
\and Nicolas Delfosse \thanks{Institut de Math\'ematiques de Bordeaux, UMR 5251, Universit\'e Bordeaux 1 --- 351, cours de la Lib\'eration --- 33405 Talence Cedex, France. {\tt nicolas.delfosse@math.u-bordeaux1.fr,} {\tt gilles.zemor@math.u-bordeaux1.fr}}
\and Gilles Z\'emor \footnotemark[2]
}

\begin{document}
\maketitle

\begin{abstract}
We study a construction of Quantum LDPC codes proposed by MacKay, Mitchison and Shokrollahi. It is based on the Cayley graph of $\F_2^n$ together with a set of generators regarded as the columns of the parity--check matrix of a classical code.
We give a general lower bound on the minimum distance of the Quantum code in $\mathcal{O}(dn^2)$ where $d$ is the minimum distance of the classical code.
When the classical code is the $[n,1,n]$ repetition code, we are able to compute the exact parameters of the associated Quantum code which are $[[2^{n}, 2^{\frac{n+1}{2}}, 2^{\frac{n-1}{2}}]]$. 
\end{abstract}

\bigskip



\bigskip

\noindent{\bf Notes.} The material in this paper was presented in part at ISIT 2011 \cite{cdz11}.
This article is published in {\em IEEE Transactions on
Information Theory} \cite{CDZ13}.
We point out that the second step of the proof of
Proposition VI.2 in the published version (Proposition~\ref{prop:wrong_proof}
in the present version and Proposition 18 in the ISIT extended abstract
\cite{cdz11}) is not strictly correct. This issue is addressed
in the present version.

\section{Introduction}

Classical LDPC codes, it hardly needs to be recalled, come together with
very efficient and fast decoding algorithms and overall display extremely
good performance for a variety of channels. Quantum error-correcting codes
on the other hand, under the guise of the CSS \cite{CS96a,Ste96b} scheme,
are in some ways strikingly similar to classical codes, and in particular
can be decoded with purely classical means. It is therefore natural to
try to import the classical LDPC know-how to the Quantum setting. There is
however a structural obstacle. A Quantum CSS code is defined by two
binary parity-check matrices whose row-spaces must be orthogonal to each other.
To have a Quantum LDPC code decodable by message-passing these two matrices
should be sparse, as in the classical case. Therefore, randomly choosing
these matrices, the generic method which works very well in the classical case,
is simply not an option in the Quantum case, because the probability of
finding two sparse row-orthogonal matrices is extremely small. A number of
constructions have been suggested by classical coding theorists nevertheless
\cite{MMM04, Al07, Al08, COT07, HMIH07, SKR08}
but they
do not produce families of Quantum LDPC codes with a minimum distance growing
with the blocklength. While this may be tolerable for practical
constructions of fixed size, this is clearly an undesirable feature
of any asymptotic construction and it raises the intriguing theoretical
question of how large can the minimum distance of sparse (or LDPC) CSS codes be.
Families of sparse CSS codes with a growing minimum distance do exist, the
most well-known of these being Kitaev's toric code \cite{Kit03a}, which has
been generalised to codes based on tesselations of surfaces (see \emph{e.g.}
\cite{BK98, FML02a, BM06, BM07a, Ze09}) and higher-dimensional objects. These constructions 
exhibit minimum distances that scale at most as a square root of the 
blocklength $N$ (to be precise, $N^{1/2}\log N$ is achieved in \cite{FML02a})
though this often comes at the cost of a very low dimension (recall
that the dimension of the toric code is $2$). It is an open question
as to whether families of sparse CSS codes exist with a minimum distance
that grows at least as $N^\alpha$ for $\alpha >1/2$, even for Quantum
codes with dimension $1$.
The recent construction \cite{tz09} manages to reconcile a minimum distance of the order of $N^{1/2}$
with a dimension linear in the blocklength. All these constructions
borrow ideas from topology and can be seen as some generalisation of Kitaev's
toric code. 

In a follow-up to the paper \cite{MMM04a}
MacKay, Mitchison and Shokrollahi \cite{MMS} proposed a construction
that seemingly owes very little to the topological approach. They noticed
that the adjacency matrix of any Cayley graph over $\F_2^r$ with an even
set of generators is self-dual and can therefore be used to define a sparse
CSS code. Experiments with some Cayley graphs were encouraging. In the present
work we take up the theoretical study of the parameters of these CSS
codes which was left open by MacKay \emph{et al.} The Quantum code in the
construction is defined by a classical $[n,k,d]$ linear binary code where
$n$ must be even.
Its length is $N=2^{n-k}$, and the row-weight of the parity-check matrix
is $n$. The dimension and the minimum distance of the Quantum code
does not depend solely on the classical code's parameters, but depend more
subtly on its structure. We solve the problem in the first non-trivial
case, which was an explicit question of MacKay \emph{et al.}, namely the case
when the classical code is the $[n,1,n]$ repetition code. Computing
the parameters of the associated Quantum code turns out to be not easy,
even in this apparently simple case. Our main result, Theorem \ref{thm:main},
gives the exact parameters for this Quantum code, namely:
$$
[[N = 2^{n}, K = 2^{\frac{n+1}{2}}, D = 2^{\frac{n-1}{2}}]].
$$
The construction therefore hits the $N^{1/2}$ barrier for the minimum distance,
but it is quite noteworthy that it does so using a construction that breaks
significantly with the topological connection. For Quantum codes based
on more complicated classical $[n,k,d]$ structures, similarly precise
results seem quite difficult to obtain, but we managed to prove a lower
bound on the Quantum minimum distance of the form $D\geq adn^2$
for some constant $a$ (Theorem \ref{thm:min_dist}). 

Notice that the constructed quantum LDPC codes have not a constant row-weight. 
Indeed, this weight is logarithmic in the blocklength.
This has its drawbacks since decoding will be slightly more complex, we
remark however that the best families of classical LDPC codes (i.e.
capacity-achieving LDPC codes) all
have row weights that grow logarithmically in the block length. We note also that it
was recently proved in~\cite{DZ}
that quantum LDPC stabilizer codes cannot achieve the capacity of the quantum
erasure channel if their stabilizer matrices have constant row weight.

\subsection*{Outline of the article} Some prerequisites on Quantum and
Quantum CSS codes together with some basic notions on Cayley graphs
are recalled in Section \ref{sec:prelim}. In Sections \ref{sec:basic}, we describe some basic properties of Cayley graphs
associated to the group $\F_2^n$.
In Section \ref{sec:Hamm}, we focus on the
properties of the {\it Hamming hypercube}, that is the Cayley graph
$\G (\F_2^n, S_n)$, where $S_n$ denotes the canonical basis. In
particular, we observe some nice property: for almost all families $S$ of generators of $\F_2^m$, the Cayley graph $\G (\F_2^{m}, S)$ looks locally like the Hamming hypercube of dimension $\# S$.
In Section \ref{sec:dist}, we study the minimum distance of a Quantum code associated to a Cayley graph of $\F_2^n$ and show that this distance is at least quadratic in $n$. Finally in Section \ref{sec:repet}, we focus on the example studied by Mitchison {\it et al.} in \cite{MMS} and give the exact parameters of this family of Quantum codes.

\section{Preliminaries}\label{sec:prelim}

In this article all codes, classical and quantum, are binary.

\subsection{Self-Orthogonal Codes and Quantum codes}



\begin{defn}
  A classical code $C\in \F_2^n$ is said to be \emph{self-orthogonal} if $C \subset C^\bot$. It is said to be self-dual if $C=C^\bot$. For convenience's sake, we also say that a binary $r\times r$ matrix $H$ is \textit{self-orthogonal} (resp. \textit{self-dual}) if $HH^T=0$ (resp. $HH^T$ and $\Rk (H)=r/2$).
\end{defn}

Classical self-orthogonal codes provide a way of constructing quantum
codes through a particular case of the CSS construction
\cite{CS96a,Ste96b}. Let us just recall that if $C$ is self-orthogonal
with classical parameters $[n,k,d]$, then it yields a quantum code with
 parameters $[[N , K, D]]$, where $N=n$, $K=n-2k$ and 
where $D$ is the minimum weight of a codeword in $C^\bot \setminus C$.

Notice that this last characterization of $D$ implies that $D\geq
d^\perp$ where $d^\perp$ denotes the dual distance of $C$. One way of
obtaining quantum codes with good parameters is therefore simply to
use classical self-orthogonal codes with a large dual
distance: this approach has been used repeatedly to obtain record
parameters. 
However, our purpose is to construct CSS codes with a
low-density stabilizer (parity-check) matrix, meaning that we need
a {\em sparse} self-orthogonal matrix~$H$. Since we have $d^\perp \leq d$ for the
self-orthogonal code $C$ generated by the rows of $H$, the bound $D\geq
d^\perp$ is of little use because it cannot bound $D$ from below by
anything more than the (low) weight of the rows of $H$. Obtaining a
better lower bound on the quantum code's minimum distance~$D$ can be quite challenging.


In the present work we shall develop a method to obtain improved lower
bounds on $D$ for some quantum codes based on sparse self-orthogonal matrices.
We focus on MacKay et al.'s construction based on the adjacency matrices of some Cayley graphs. Let us first recall some basic notions on Cayley Graphs.

\subsection{Cayley graphs and CSS codes}

\subsubsection{The general construction}

\begin{defn}
  Let $G$ be a group and $S$ be a subset of $G$. The Cayley graph $\G (G,S)$ or $\G (S)$, when there is no possible confusion, is the graph whose vertex-set equals $G$ and such that two vertices $g, g'\in G$ are connected by an edge if there exists $s\in S$ such that $gs=g'$. 
\end{defn}

\begin{rem}\label{rem:oriented_symmetric}
  The graph $\G (G, S)$ is oriented unless $S^{-1}= S$. In addition, if $S^{-1} =S$, then, the adjacency matrix of the graph is symmetric.
\end{rem}

\begin{rem}
The graph $\G (G, S)$ is connected if and only of $S$ generates $G$.  
\end{rem}

Our point is to get pairs $(G,S)$ such that the adjacency matrix $H$ of $\G (G, S)$ is self--orthogonal, i.e. such that $HH^T=0$.
Notice that  $HH^T=0$  happens if and only if both conditions are satisfied.

\begin{enumerate}[(1)]
\item\label{item:row_self} Each row of $H$ is self-orthogonal, i.e. has even weight;
\item\label{item:row_orth} Any pair of distinct rows of $H$ are orthogonal, i.e. any two distinct rows of $H$ have an even number of $1$'s in common.
\end{enumerate}

The following proposition translates the above conditions in terms of the pair $(G,S)$.

\begin{prop}\label{prop:How_to}
  Let $G$ be a finite group and $S$ be a system of generators of $G$. Assume that 
  \begin{enumerate}[(i)]
  \item\label{item:even} $\# S$ is even;
  \item\label{item:expressions} for all $g\in G$, there is an even number of distinct expressions of $g$ of the form $g=st^{-1}$, with $(s,t)\in S^2$.
  \end{enumerate}
Then, the adjacency matrix of the Cayley graph $\G (G, S)$ is self-orthogonal.
\end{prop}

\begin{proof}
  Condition (\ref{item:even}) entails obviously (\ref{item:row_self}).
Now, let $a,b$ be two distinct elements of $G$ and $H_a, H_b$ the corresponding rows of the adjacency matrix of $\G (G, S)$.
The rows have a $1$ in common if and only if $at=bs$ for some pair $(s,t) \in S^2$. This equality is equivalent with $b^{-1}a=st^{-1}$. Thus, (\ref{item:expressions}) naturally entails (\ref{item:row_orth}).
\end{proof}

\begin{rem}
  If $S^{-1}=S$, then the graph in undirected, its adjacency matrix is symmetric and (\ref{item:expressions}) can be replaced by 
  \begin{enumerate}
  \item[(ii')]  for all $g\in G$, there is an even number of distinct expressions of $g$ of the form $g=st$, with $(s,t)\in S^2$.
  \end{enumerate}
\noindent It is worth noting that if $s$ and $t$ commute, then $g=st$ and $g=ts$ correspond to distinct expressions.
\end{rem}

\subsubsection{The group algebra point of view}
We still consider a pair $(G,S)$, where $G$ is a group and $S$ is a generating set of $G$. Recall that the \emph{group algebra of $G$ over $\F_2$} denoted by $\F_2[G]$ is the $\F_2$--vector space with a basis $\{e_g,\ g\in G\}$ in one--to--one correspondence with elements of $G$ together with a multiplication law induced by the group law, i.e. $e_g.e_{g'}=e_{gg'}$.

\begin{nota}\label{nota:pi_s}
Given a pair $(G, S)$, where $G$ is a group and $S$ a generating set.
We denote respectively by $\pi_S$ and $\hat{\pi}_S$ the elements of $\F_2 [G]$,
$$
\pi_S:=\sum_{s\in S} s \qquad \textrm{and}\qquad \hat{\pi}_S:= \sum_{s\in S} s^{-1}.
$$
Clearly, the two elements are equal when $S=S^{-1}$.
\end{nota}

\begin{lem}\label{lem:mult_pi}
  The adjacency matrix $H$ of $\G (G, S)$ represents the right multiplication by $\pi_S$ i.e. the application
$$
\phi_S: \left\{
  \begin{array}{ccc}
    \F_2 [G] & \rightarrow & \F_2 [G] \\
     f & \mapsto & f \pi_S
  \end{array}
\right. .
$$
In addition, the matrix $H^T$ represents the right multiplication by $\hat{\pi}_S$.
\end{lem}


\noindent {\bf Caution.} In Lemma \ref{lem:mult_pi} above, we suppose that matrices act on row-vectors, i.e. an $n\times n$ binary matrix $M$ corresponds to an endomorphism of $\F_2^n$ by $v\mapsto vM$, where $v\in \F_2^n$ is represented by a row-vector.

\begin{lem}\label{lem:get_nilp}
  The adjacency matrix $H$ of $\G (G,S)$ is self--orthogonal if and only if $\pi_S \hat{\pi}_S=0$. In particular, if $S=S^{-1}$, then $H$ is self-orthogonal if and only if $\pi_S^2=0$.
\end{lem}


In particular, the problem of finding sparse self-orthogonal matrices is equivalent with that of finding a $2$-nilpotent element of $\F_2[G]$ having a low weight compared to $2^{\# G}$. 

\subsection{Some examples}

\begin{ex}
Let $G$ be the group $(\Z/2n\Z)^2$ and $S$ be the set $S:=\{(1,0),(0,1), (-1,0), (0,-1), (n+1,0), (n-1,0), (0, n+1), (0, n-1)\}$. Then, the adjacency matrix of $\G ((\Z/2n\Z)^2, S)$ is self-orthogonal. The corresponding group algebra is isomorphic to $\F_2[x,y]/(x^{2n}-1, y^{2n}-1)$ and the element $\pi_S$ equals $x+y+x^{n-1}+y^{n-1}+y^{n+1}+x^{n+1}+x^{2n-1}+y^{2n-1}$.
\end{ex}

Motivated by MacKay et al.'s draft \cite{MMS}, the group we will focus
on in the rest of the paper is $G=\F_2^n$. Since we are dealing with
an abelian group we denote group operations additively rather than multiplicatively.

\begin{ex}\label{ex:ex_cube}
  $G=\F_2^n$ and $S$ is any system of generators with an even number of elements.
The corresponding group algebra is isomorphic to $\F_2[x_1, \ldots , x_n]/(x_1^2-1, \ldots , x_n^2-1)$, in which one sees easily that any element of even weight satisfies $f^2=0$.
\end{ex}

\section{Basic properties of CSS codes from Cayley Graphs of $\F_2^n$}\label{sec:basic}

As we have just seen in the last example, any even number of generators of $\F_2^n$
defines a Cayley graph whose adjacency matrix is a
 self-orthogonal matrix, from which we have a quantum code. The row
 weight of the matrix is equal to the cardinality of the set of
 generators: when this cardinality is chosen proportional to $n$, we
 have a row weight that is logarithmic in the row length, hence the
 LDPC character of the quantum code. As put forward in \cite{MMS},
 note also that the matrix is $2^n\times 2^n$, i.e. has a highly
 redundant number of rows, which is beneficial for decoding. It also
 makes the computation of its rank, and hence the dimension of the
 quantum code, non-trivial. The present paper strives to compute or estimate
 parameters, dimension and minimum distance, of the resulting quantum
 LDPC code.

\subsection{Context and notation}

One of the main difficulties of the following work is that we juggle with different kinds of classical codes. Roughly speaking, we deal with {\it small} codes of length $n$ and {\it big} codes of length $2^n$.

This is the reason why we first need to describe carefully the landscape and the notation we choose.

\subsubsection{The ``small'' and ``big'' objects}\label{subsubsection:small_and_big}
For a positive integer $n$, the canonical basis of $\F_2^n$ is denoted by $S_n:=(e_1, \ldots , e_n)$.
In what follows, words of $\F_2^n$ are denoted by letters in lower case such as $c, m$ or $x$.
Such words are referred as {\it small words} and subspaces of $\F_2^n$ are referred as {\it small codes}.

Given a set $S$ of generators of $\F_2^n$ we denote by $\M (\F_2^n, S)$, or $\M (S)$ when no confusion is possible, an adjacency matrix of the Cayley graph $\G (\F_2^n , S)$.
From Proposition \ref{prop:How_to}, if $\# S$ is even, then $\M (S)$ is self--orthogonal.
We denote by $\C (\F_2^n , S)$ or $\C (S)$ the code with generator matrix $\M (S)$.
Words of this code or more generally of its ambient space, namely $\F_2^{2^n}$ will be denoted by letters in Gothic font such as $\c$ or $\d$.
In what follows and to help the reader, we frequently refer to {\it big words} and {\it big codes} when dealing with such words or codes. Gothic fonts are dedicated to {\it big} objects, such as the matrices $\M (S)$, the Cayley graphs $\G(S)$, the corresponding big codes $\C (S)$ and so on...

\subsubsection{Graphs}
In a graph $\G$, we say that two connected vertices have distance $r$
if the the shortest path between them consists of $r$ edges. This
defines a natural metric on $\G$. 

\begin{nota}
For this distance, a ball centred at
a vertex $x$ of radius $\rho$ is denoted by $\B(x, \rho)$, it is the
set of vertices at distance $\leq \rho$ of $x$. A sphere of centre $x$
and radius $\rho$ is denoted by $\Sp (x, \rho)$.  
\end{nota}

We will say that a graph $\H$ is a {\em cover} or a {\em lift} of $\G$
if it comes together with a surjective map $\gamma : \H \rightarrow \G$ called a
{\em covering map} such that for any vertex $h$ of $\H$, the map
$\gamma$, restricted to the set of neighbours of $h$, is a one-to-one
mapping onto the set of neighbours
of $\gamma(h)$. The covering map $\gamma$ is a local isomorphism.
It can be shown that when $\G$ is connected, the
cardinality of the preimage of any vertex is constant: we will refer
to this number as the {\em degree} of the cover.

Consider the particular case when $\G = \G (\F_2^m, T)$ for $T$ some
set of generators of $\F_2^m$. A natural covering map of $\G$ is 
\begin{equation}
  \label{eq:cover}
 \gamma: \H = \G (\F_2^{\# T}, S_{\# T}) \longrightarrow \G = \G (\F_2^m, T)
\end{equation}
which can be thought of as removing linear dependencies between
elements of $T$ (see \S \ref{section:hyper_cover}). Any Cayley graph associated to $\F_2^m$ is therefore locally
isomorphic to some
hypercube $\G (\F_2^n , S_n)$.
This covering construction was used by Tillich and Friedman in
\cite{FT02}. Starting with a code $C$ of generating matrix $M$, they
used the set $T$ of columns of $M$ to define a graph $\G = \G (\F_2^m, T)$:
relating the eigenvalues of $\G$ to those of its cover
\eqref{eq:cover} they derived upper bounds on the minimum distance of
$C$.
Here we shall rather view the set of generators $T$ as the set of
columns of a code $C$'s parity-check matrix (rather than a generating
matrix). The minimum distance $d$ of $C$ is therefore the minimum
weight of a
linear relation between generators of $T$, and for $\rho<d$ the balls
 $\B(x, \rho)$ in $\H$ and $\B(\gamma(x), \rho)$ in $\G$ are
isomorphic.

\subsubsection{The dictionary relating big codes and graphs}

We keep the notation of \S \ref{subsubsection:small_and_big}. 
It is worth noting that elements of the ambient space of $\C(S)$ are
in one--to--one correspondence with subsets of the vertex--set of $\G
(S)$. In what follows, we frequently allow ourselves to regard big
words as sets of vertices, while vertices are nothing but elements of
$\F_2^n$. In particular we allow ourselves notation such as ``$x\in \c$'', where $x\in \F_2^n$ and $\c\in \F_2^{2^n}$. 
From this point of view, we frequently use the elementary lemma below.
Recall that, given two subsets $A, B$ of a set $E$, the symmetric difference of $A$ and $B$ is defined by $A \bigtriangleup B:=(A\cup B) \setminus (A\cap B)$. This operation is associative. 

\begin{lem}\label{lem:dict}
Regarding elements of the ambient space of $\C(\F_2^n, S)$ as subsets of the vertex--set of $\G (\F_2^n, S)$,

\begin{enumerate}[(1)]
  \item a row of $\M(\F_2^n, S)$ is nothing but a sphere $\Sp (x, 1)$ of centre $x\in \F_2^n$ and radius $1$, where $x$ is the index of the row;
  \item a word of $\C (\F_2^n, S)$ is a symmetric difference of spheres of radius $1$, or equivalently an $\F_2$--formal sum of such spheres;
  \item\label{item:geo_dual} a word $\c \in \C (\F_2^n, S)^{\bot}$ is a set of vertices such that for every sphere $\Sp(x, 1)$ of radius $1$, the intersection $\c \cap \Sp (x, 1)$ has even cardinality. 
\end{enumerate}  
\end{lem}

\subsection{Automorphisms of the big codes and the graphs}

Given a positive integer $n$, recall that the Hamming--isometries $\phi: \F_2^n \longrightarrow \F_2^n$ are of the form $\phi=\sigma \circ t_m$, where $\sigma$ is a permutation of the coordinates and $t_m$ is the affine translation $x\longmapsto x+m$ for some fixed $m\in \F_2^n$.

\begin{lem}
  Let $S$ be a family of generators of $\F_2^n$ and $\phi$ be a Hamming--isometry of $\F_2^n$, then $\phi$ induces a permutation $\Phi$ of $\F_2^{2^n}$ which is an automorphism of $\G (S)$ and an element of the permutation group of $\C(S)$ (and hence in that of $\C(S)^{\bot}$).
\end{lem}

\begin{proof}
For all small word $x\in \F_2^n$, the sphere $\Sp (x, 1)$ is the big word whose nonzero entries are the small words $x+s$ with $s\in S$.
  The code $\C(S)$ is generated by the $\Sp (x, 1)$'s for $x\in\F_2^n$ and one sees easily that $\Phi (\Sp (x, 1))=\Sp({\phi (x)}, 1)$.
\end{proof}

\begin{cor}\label{cor:vertex_trans}
Let $m\in \F_2^m$, if there exists a nonzero big word $\c$ in $\C(S)$ (resp. $\C(S)^{\bot}$), then, there exists a big word $\c' \in \C(S)$ (resp. $\C(S)^{\bot}$) with the same weight and which contains the small word $m$.
\end{cor}

\section{The Hamming hypercube}\label{sec:Hamm}

In this section, $n$ denotes an even integer and we study the properties the Cayley graph $\G (\F_2^n, S_n)$.
Recall that $S_n$ denotes the canonical basis of $\F_2^n$.

First, we show that $\Rk (\M( S_n))=2^{n-1}$, which means that the corresponding big code is self--dual and hence that the corresponding CSS code is trivial.
However, the properties of $\G(\F_2^n, S_n)$ are of interest because
of its role in
the covering construction \eqref{eq:cover}.


\subsection{The corresponding Quantum code is trivial}

\begin{prop}\label{prop:Rk_An}
Let $n$ be an even integer. The adjacency matrix $\M( S_n)$ of $\G (\F_2^n, S_n)$, satisfies
$$
\Rk (\M(S_n))=2^{n-1}.
$$
Therefore, $\M (S_n)$, or equivalently $\C (S_n)$, is self-dual.
\end{prop}

\begin{proof}
The group algebra of $\F_2^n$ is $\F_2[\F_2^n]\simeq \F_2 [X_1, \ldots , X_n]/(X_1^2 -1, \ldots , X_n^2 -1)$. Using {\it Notation}~\ref{nota:pi_s}, the element $\pi_{S_n}$ is $X_1+\cdots +X_n$.
Thus, the cokernel of the endomorphism $\phi_{S_n} : x \longrightarrow x\pi_{S_n}$ is  $$\coker \phi_{S_n} = \F_2 [X_1, \ldots , X_n]/(X_1^2 -1, \ldots , X_n^2 -1, X_1+\cdots +X_n).$$ This last algebra is isomorphic to $\F_2 [X_1, \ldots , X_{n-1}]/(X_1^2 -1, \ldots , X_{n-1}^2 -1, (X_1+\cdots +X_{n-1})^2-1)$ and one sees easily that if $n$ is even, then $(X_1+\cdots +X_{n-1})^2-1= X_1^2 -1+ \cdots + X_{n-1}^2 -1$. Thus, this cokernel is isomorphic to 
$\F_2[X_1, \ldots , X_{n-1}]/(X_1^2 -1, \ldots , X_{n-1}^2 -1) \cong \F_2[\F_2^{n-1}]$ whose $\F_2$--dimension is exactly the half of that of $\F_2[\F_2^n]$.  
\end{proof}

\subsection{The graph is bipartite}

Another very useful and nice property of this family of graphs is given by the following statement.

\begin{prop}
  \label{prop:bipartite}
Consider the partition of $\F_2^n$ by $G_{even} \cup G_{odd}$ of small words of having respectively even and odd Hamming weight. Then, $\G (\F_2^n, S_n)$ is bipartite, i.e. any edge links an element of $G_{even}$ with one of $G_{odd}$.
\end{prop}

\begin{proof}
For all $x\in \F_2^n$ and all $e_i\in S_n$, the small words $x$ and
$x+e_i$ have weights of distinct parities.  
\end{proof}

\begin{rem}
  In matrix terms, this means, that, for a suitable ordering of the elements of $\F_2^n$, there exists a $2^{n-1}\times 2^{n-1}$ binary matrix $U_n$ such that
\begin{equation}\label{eq:Mat}
\M(S_n)=\left(
  \begin{array}{ccc}
    (0) & & U_n \\
        & &     \\
    U_n^T & & (0)
  \end{array}
\right).
\end{equation}
In addition, one shows easily by induction on $n$ that $U_n^T=U_n$. 
\end{rem}

The former result has interesting consequences on the code $\C(S_n)$ for $n$ even.

\begin{cor}
 Let $n$ be an even integer. The code $\C (\F_2^n, S_n)$ splits in a direct sum of two isomorphic subcodes with disjoint supports
$$
\C (S_n) = C (S_n)_{even} \oplus C (S_n)_{odd}
$$
corresponding to big words whose supports are the small words of even and odd weight respectively. Both subcodes are self--dual.
\end{cor}

\begin{proof}\label{cor:split_code}
  The two codes come respectively from the upper and lower halves of the row-set of $\M (S_n)$ in (\ref{eq:Mat}).
They are obviously isomorphic since they have the same generator matrix $U_n$.
The self--orthogonality is clear since $\M(S_n)\M(S_n)^T=0$ entails $U_nU_n^{T}=0$. In addition, it is clear that $\Rk (U_n)=\frac{1}{2}\Rk \M (S_n)=2^{n-2}$, which yields self--duality.
\end{proof}

\begin{prop}
  Using the notation of Proposition \ref{prop:bipartite} and Corollary \ref{cor:split_code}, a big word $\c \in \F_2^{2^{n}}$ whose support is contained in $G_{even}$ (resp. $G_{odd}$) is in $\C (S_n)_{even}$ (resp. $\C (S_n)_{odd}$) if and only if it is orthogonal to any sphere $\Sp (x,1)$ where $x$ is a small word of odd (resp. even) weight.
\end{prop}

\begin{proof}
 Since $\C(S_n)$ is self-dual, a big word is in $\C(S_n)$ if and only if it is orthogonal to any sphere of radius $1$. If $x$ is a small word of even weight, then the elements of $\Sp (x, 1)$ have odd weight and hence is obviously orthogonal to any big word supported in $G_{even}$. Thus, a big word with support in $G_{even}$ (resp. in $G_{odd}$) is in $\C (S_n)$ if and only if it is orthogonal to any sphere of radius $1$ centred at a small word of even (resp. odd) weight.
\end{proof}

Consequently, the graph $\G (S_n)$ can be regarded as a Tanner graph for $\C (S_n)_{even}$ where $G_{even}$ is the set of bit nodes and $G_{odd}$ the set of check nodes. It can conversely be regarded as a Tanner graph for $\C (S_n)_{odd}$ by switching bit and check nodes. 

\begin{figure}[!h]
  \centering
  \includegraphics[scale=0.5, angle=90]{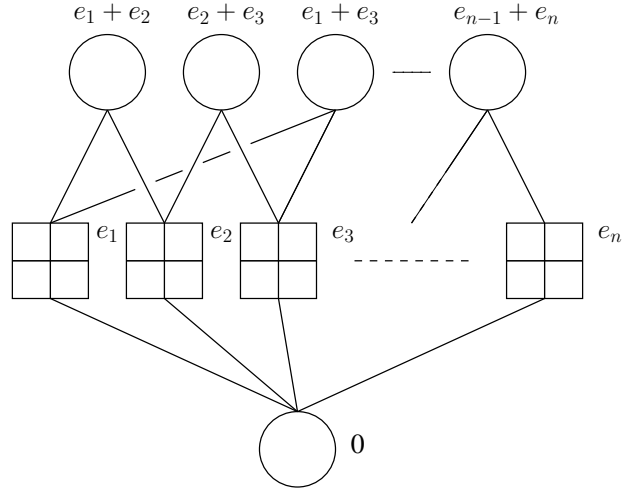}
  \caption{A part of the Hamming cube regarded as a Tanner Graph}
  \label{fig:Tanner}
\end{figure}

\begin{rem}\label{rem:bipart}
  Actually, this property of being bipartite is satisfied by any Cayley graph $\G (\F_2^n, S)$ as soon as for all $x\in \F_2^n$ and all $s\in S$, the weights of the small words $x$ and $x+s$ have distinct parities. It holds for instance for $\G (\F_2^m, S_m\cup \{e_1+\cdots +e_m\})$, where $m$ is odd.
\end{rem}

\subsection{A property of bounded codewords}

The following statement is crucial in the study of the minimum distance of Quantum codes from graphs covered by $\G (S_n)$.

\begin{prop}\label{prop:local_sum}
  Let $\c$ be a codeword in the row-space of $\M (S_n)$. Regarding $\c$ as a subset of the vertex--set of $\G (S_n)$, assume that $\c$ is contained in the ball $\B (x, r)$ for some vertex $x\in \F_2^n$ and some integer $r<n$. Then $\c$ is a sum of rows of $\M (S_n)$ with support contained in  $\B (x, r)$. Equivalently, $\c$ is the $\F_2$--formal sum of spheres of radius $1$ contained in $\B (x, r)$.
\end{prop}

\begin{proof}
From Corollary \ref{cor:vertex_trans}, one can assume that $x=0$ and hence $\c\subset \B (0, r)$. 
Let us prove the result by induction on $r$.

We will consider the extremal points of $\c$, that is the vertices of $\c$ whose distance $r$ to 0 is maximal. For every extremal vertex $v$ of $\c$, we will add a sphere included in the ball $\B(0, r)$ to $\c$ to obtain a new codeword $\c'$ which does not contain the vertex $v$. This procedure will lead to a decomposition of $\c$ as a sum of spheres included in the ball $\B(0, r)$.

\medbreak

If $r=0$, then $\c$ is either the zero codeword or the unique big word with support equal to the vertex $0\in \F_2^n$. But the big word of $\F_2^{2^n}$ with support equal to the vertex $0$ cannot be in $\C (S_n)$. Indeed, this big word has weight $1$ and since $\M (S_n)$ is self--dual, if it had such a big word in its row--space, the word would lie in its kernel. Thus, $\M (S_n)$ would have a zero column which is impossible.
Thus, $\c$ is the big word zero which is the empty formal sum of spheres of radius $1$.

\medbreak

Let $r>0$ and assume that the result holds for all radius $r'<r$.

\medbreak

\noindent {\it Claim.} Let $\rho \leq r$ be the least integer such that $\c \subseteq \B(0, \rho)$. If $\c\neq \emptyset$, then, clearly, $\c \cap \Sp (0, \rho)$ is nonempty. Then, for all $i\in \{1, \ldots , n\}$, there exists $c\in \c \cap \Sp(0, \rho)$ whose $i$-th entry is nonzero.

\medbreak

\noindent {\it Proof of the claim.} Assume the claim is false. Without loss of generality, one can assume that the $n$--th entry of any element of $\c \cap \Sp (0, \rho)$ is zero. Thus, the elements of $\c \cap \Sp (0, \rho)$ are of the form $(m_1|0), \ldots , (m_s | 0)$, where the $m_i$'s $\in \F_2^{n-1}$ and the ``$|$'' denotes the concatenation. 
From, Proposition \ref{prop:Rk_An}, we have $\C (S_n)=\C(S_n)^{\bot}$. Thus, regarding $\c$ as an element of $\C(S_n)^{\bot}$ and using Lemma \ref{lem:dict}(\ref{item:geo_dual}), we see that the intersection of $\c$ with any sphere of radius $1$ has an even cardinality. However, the spheres $\Sp((m_i|1), 1)$ contain one and only one element of $\c$, namely $(m_i|0)$. This yields the contradiction.

\medbreak

Thanks to the claim, we know that there exists at least one element of $\c \cap \Sp (0, \rho)$ with a nonzero $n$--th entry. Let $(\ell_1|1), \ldots , (\ell_t|1)$ be these elements. Clearly, the small words $\ell_i\in \F_2^{n-1}$ have weight $\rho-1$ and hence the spheres $\Sp((\ell_i|0),1)$ are contained in $\B(0, \rho)$. For all $i$, the only element of $\Sp ((\ell_i|0), 1) \cap \Sp (0, \rho)$ whose $n$--th entry is nonzero is $(\ell_i|1)$. Thus, the big word
\begin{equation}\label{eq:d_c}
\d:=\c +\Sp((\ell_1|0),1)+\cdots + \Sp((\ell_t|0),1)
\end{equation}
is contained in $\B(0, \rho)$ and the elements $\d \cap \Sp (0, \rho)$ have all a zero $n$--th entry. Indeed, the $(\ell_i|1)$'s have been cancelled and no other element of the form $(\ell|1)$ have been added while adding the spheres of radius $1$. The claim entails that $\d\subset \B(0, \rho-1)$. By the induction hypothesis, $\d$ is a sum of spheres of radius $1$ contained in $\B (0, \rho -1)$. Since the spheres $\Sp((\ell_i|0),1)$ are also contained in $\B(0, \rho)$, Equation (\ref{eq:d_c}) yields the result. 
\end{proof}

\subsection{The hypercube cover}\label{section:hyper_cover}

\begin{nota}
  In what follows, $m$ denotes an integer. Recall that $S_m$ denotes the canonical basis of $\F_2^m$. Let $W$ be a family of distinct nonzero elements of $\F_2^m \setminus S_m$ with cardinality $w:=\# W$ and assume that $m+w$ is even.
From Proposition \ref{prop:How_to}, the code $\C(S_m \cup W)$ is
self-orthogonal and hence provides a Quantum CSS code with parameters
$[[2^m, 2^m -2\dim \C(S_m\cup W), D ]]$, where $D$ is the minimum
weight of a codeword of $\C(S_m \cup W)^{\bot}\setminus \C (S_m \cup W)$. 
\end{nota}

Regarding the elements of $S_m \cup W$ as column vectors, we introduce the binary $m\times (m+w)$ matrix $M(W)$ whose columns correspond to the elements of $S_m\cup W$, that is
\begin{equation}\label{eq:MS}
M(W):=\left( \left. \begin{array}[c]{ccc}
       & & \\ & I_m & \\ &  &  
    \end{array} \right|
    \ P(W)\ 
 \right),
\end{equation}
where $I_m$ denotes the $m\times m$ identity matrix and $P(W)$ is the matrix whose columns are the elements of $W$.

\begin{thm}\label{thm:d_ball}
  Let $C(W)$ be the code with parity--check matrix $M(W)$. There is a natural graph cover $$\gamma_{W} : \G (\F_2^{m+w}, S_{m+w}) \longrightarrow \G(\F_2^m, S_m\cup W).$$ The degree of $\gamma_{W}$ is $\# C(W)$. In addition, denoting by $d$ the minimum distance of $C(W)$, the restriction of $\gamma_W$ to any ball of radius $\leq \lfloor \frac{d-1} 2 \rfloor$ is an isomorphism of graphs. 
\end{thm}

\begin{proof}
  Recall that we denote the elements of the canonical basis $S_m$ by $e_1, \ldots , e_m$. Denote by $e_1', \ldots e_w'$ the elements of $W$. Consider the linear map
$$
\left\{
  \begin{array}{ccc}
    \F_2^{m+w} & \longrightarrow & \F_2^m \\
    x & \longmapsto & M(W).x^T
  \end{array}
\right. ,
$$
that sends $e_1\mapsto e_1, \ldots, e_m\mapsto e_m, e_{m+1}\mapsto
e_1',\ldots , e_{m+w}\mapsto e_w'$. The covering map $\gamma_{W}$ is
naturally constructed from the above map. One sees easily that the
fibre (preimage) of a vertex $c$ of $\G (\F_2^m, S_m\cup W)$ is nothing but the coset $c+C(W)$ and hence has cardinality $\# C(W)$.

To conclude, consider a ball of $\G(\F_2^{m+w}, S_{m+w})$ of radius $\leq \lfloor \frac{d-1} 2 \rfloor$.
Notice that two vertices $x,x'\in \F_2^{m+w}$ of $\G(\F_2^{m+w}, S_{m+w})$ have the same image by $\gamma_{W}$ if and only if $Mx^T=Mx'^T$, i.e. if and only if $x'=x+c$ with $c\in C(W)$. In particular, two such vertices have the same image only if their distance is $\geq d$. 
Since the distance between any two vertices in a ball of radius $\leq \lfloor \frac{d-1} 2 \rfloor$ is $<d$ then, they have distinct images by $\gamma_{W}$. Thus, the restriction of $\gamma_{W}$ to the ball is an isomorphism.
\end{proof}

\section{On the minimum distance of the Quantum code}\label{sec:dist}

We keep the notation of \S \ref{section:hyper_cover}.
Given a set of generators $S_m\cup W$ of $\F_2^m$ as before, our point is to bound below the minimum distance of the corresponding CSS Quantum code, that is the minimum weight of the set $\C(\F_2^m, S_m\cup W)^{\bot}\setminus \C(\F_2^m, S_m\cup W)$.

\begin{prop}\label{prop:unbound}
We keep the notation of Theorem \ref{thm:d_ball}.
A codeword in $\C (S_m \cup W)^{\bot} \setminus \C (S_m \cup W)$ is not contained in a ball of radius $\lfloor \frac{d-1}{2} \rfloor -2$.
\end{prop}

\begin{proof}
First, let us quickly sketch this proof. Set $t:=\lfloor \frac{d-1}{2} \rfloor$.
Assume that $\c \in \C (S_m\cup W)^{\bot} \setminus \C (S_m \cup W)$ is contained in a ball of radius $t-2$. Then, using Theorem \ref{thm:d_ball}, we construct a {\it lift} $\c^{\star}$ of $\c$ satisfying 
\begin{enumerate}[(1)]
\item\label{item:orth_c} $\c^{\star} \in \C(\F_2^{m+w},S_{m+w})^{\bot}$;
\item\label{item:bound_c} $\c^{\star}$ is contained in a ball of radius $t-2$;
\item\label{item:gamma_c} $\gamma_{W}(\c^{\star})=\c$, where
  $\gamma_{W}$ is the graph covering map introduced in Theorem \ref{thm:d_ball}. 
\end{enumerate}
From Proposition \ref{prop:Rk_An}, the code $\C(\F_2^{m+w}, S_{m+w})$
is self--orthogonal and hence $\c^{\star} \in \C (\F_2^{m+w},
S_{m+w})$. From Proposition \ref{prop:local_sum}, $\c^{\star}$ is a
sum of spheres contained in the ball of radius $t-2$. From Theorem~\ref{thm:d_ball}, the covering map
$\gamma_{W}$ restricted to a ball of radius $\leq t$ is an isomorphism. Thus, $\c$ is a sum of spheres and hence is a codeword of $\C (\F_2^m, S_m \cup W)$ which leads to a contradiction.

The non-obvious part of the proof is the construction of the lift $\c^{\star}$. It is worth noting that, despite $\gamma_{W}$ inducing an isomorphism between balls of radius $t$, it is however not possible to lift all such big words in a ball of radius $>t-2$. A counter-example is given in Example \ref{ex:non_lift}.

Let us prove the existence of such a lift. Without loss of generality, one can assume that $\c$ is contained in the ball $\B (0, t-2)$. Clearly $\gamma_{W}$ induces an isomorphism between this ball and the ball centred at zero of radius $t-2$ of $\G (\F_2^{m}, S_m \cup W)$.
Let $\c^{\star}$ be the inverse image of $\c$ by this isomorphism. The above conditions (\ref{item:bound_c}) and (\ref{item:gamma_c}) are obviously satisfied. There remains to prove that $\c^{\star}$ has an even number of common elements with any sphere of radius $1$. Since $\c^{\star} \subset \B (0, t-2)$, any sphere which is not contained in $\B (0, t)$ has an empty intersection with $\c^{\star}$. On the other hand, any sphere of radius $1$ contained in $\B (0, t)$ corresponds thanks to $\gamma_{W}$ and Theorem \ref{thm:d_ball} to a unique sphere of radius $1$ contained in the ball of radius $t$ centred at $0$ of $\G (\F_2^{m}, S_m \cup W)$. Thanks to this ball-isomorphism and by definition of $\c$, it is clear that $\c^{\star}$ has an even number of common elements with such a sphere. This yields (\ref{item:orth_c}).
\end{proof}

\begin{ex}\label{ex:non_lift}
  Suppose that $m=5$ and $W=\{(1 1 1 1 1)\}$. Theorem \ref{thm:d_ball}
  asserts the existence of a graph covering map $\gamma: \G(\F_2^6, S_6) \longrightarrow \G (\F_2^5, S_5 \cup W)$. The classical code $C(W)$ defined in Theorem \ref{thm:d_ball} is nothing but the pure repetition code of length $6$ which has minimum distance $6$. Therefore, $\gamma$ induces isomorphisms between balls of radius $2$. Let us show that some big words $\c\in \C (\F_2^5, S_5\cup W)^{\bot}$ contained in $\B (0, 2)$ in $\G(\F_2^5, S_5 \cup W)$ cannot lift as in the previous proof as a word $\c^{\star} \in \C(\F_2^6, S_6)^{\bot}=\C(\F_2^6, S_6)$.
Let
$$\c:=\{x\in \F_2^5,\ wt(x)=2\}.$$
Let us show that $\c \in \C(\F_2^5, S_5\cup W)^{\bot}$. Let $m\in \F_2^5$, we have to prove that $\Sp (m, 1)$ is orthogonal to $\c$, that is has an even number of common elements with $\c$. Notice that for the graph $\G(\F_2^5, S_5\cup W)$, we have 
$$
\forall m\in \F_2^5,\ \Sp(m,1):=\{m+e_1, \ldots, m+e_5, m+e_1+\cdots +e_5\}.
$$
It is clear that if $wt(m)=0, 2, 4$ or $5$, then $\Sp(m, 1)$ contains no element of weight $2$ and hence is obviously orthogonal to $\c$.
If $wt(m)=1$, then $m=e_i$ for some $i$ and $\Sp(m,1)$ contains four elements of weight $2$, namely all the $e_i+e_j$ with $i\neq j$. If $wt(m)=3$, that it $m=e_i+e_j+e_k$ for $i,j,k$ distinct to each other, then $\Sp(m, 1)$ contains also four elements of weight $2$, namely $e_i+e_j, e_i+e_k, e_j+e_k$ and $e_s+e_{\ell}$, where $\{s, \ell\}=\{1, \ldots, 5\}\setminus \{i,j,k\}$. 

Thus, $\c\in \C(\F_2^5, S_5\cup W)^{\bot}$.
From Theorem \ref{thm:d_ball}, the map $\gamma$ induces an isomorphism from the ball $\B(0,2)$ of $\G (\F_2^6, S_6)$ and the ball $\B (0, 2)$ of $\G(\F_2^5, S_5\cup W)$. 
Let us consider the lift of $\c$ by this isomorphism
$$
\c^{\star}:=\sum_{m\in \c} (m|0)\ \in \F_2^6.
$$
This big word is not an element of $\C(\F_2^6, S_6)^{\bot}$ (which equals $\C (\F_2^6, S_6)$). Indeed, let $m:=e_1+e_2+e_3\in \F_2^6$ be a small word of weight $3$. Then, $\Sp (m, 1)$ has exactly three common elements with $\c^{\star}$, namely $e_1+e_2, e_1+e_3$ and $e_2+e_3$. Thus, $\c^{\star}$ is not orthogonal to this sphere of radius $1$.
\end{ex}

\begin{thm}\label{thm:min_dist}
  Let $W$ be a family of $w$ vectors of $\F_2^m$ with $w>0$. Let
  $M(W)$ be as in (\ref{eq:MS}) and $C(W)$ be the small code of length
  $n:=m+w$ and parity--check matrix $M(W)$. Recall that $S_m$ denotes the
  canonical basis of $\F_2^m$. Let $d\geq 9$ be the minimum distance of $C(W)$. Then, the minimum distance $D$ of $\C (S_m\cup W)^{\bot} \setminus \C(S_m\cup W)$ and hence of the corresponding Quantum code satisfies
$$
D\geq \frac{1}{640} dn^2.
$$
\end{thm}

This theorem is proved further thanks to the following technical lemma.

\begin{lem}\label{lem:card_ball}
  Let $m$ and $W$ be as in Theorem \ref{thm:min_dist}. Assume moreover that the minimum distance of the small code $C(W)$ is at least $9$.
Let $\c$ be a big word of minimum weight in $\C (S_m \cup W)^{\bot}\setminus \C(S_m \cup W)$ and let $x\in \c$ then
$$
wt(\c \cap \B (x, 4))\geq \frac{n^2}{32} \cdot
$$
\end{lem}

\begin{proof}
From Corollary \ref{cor:vertex_trans}, one can assume that the small word $x$ of the statement is $0$.
From Theorem \ref{thm:d_ball}, the ball $\B(0, 4)$ of $\G (S_m\cup W)$ is isomorphic to that of $\G (S_{m+w})$. Therefore, as soon as we stay inside $\B (0,4)$, we can reason as if we were inside that of $\G (S_{m+w})$. Thus, set $n:=m+w$ and let us reason in $\G (S_n)$.

\medbreak

\noindent {\it Step 1.} First, it is important to notice that $\c$ is supposed to have a minimum weight as a word of $\C (S_m \cup W)^{\bot} \setminus \C (S_m \cup W)$, therefore,
\begin{equation}
  \label{eq:weight_c}
  \forall x\in \F_2^m,\quad wt(\c+\Sp (x,1))\geq wt(\c).
\end{equation}

\medbreak

\noindent {\it Step 2.} Since by assumption $0\in \c$ and $\c \in \C(S_m\cup W)^{\bot}$, this word must be orthogonal to any sphere of radius $1$. Denote by $e_1, \ldots , e_n$ the elements of the canonical basis $S_n$ of $\F_2^n$. Then $\c$ must be orthogonal to the spheres $\Sp (e_i, 1)$. Thus, 
\begin{equation}
  \label{eq:exists}
  \forall i \in \{1, \ldots, n\},\qquad \exists j\neq i,\ {\rm such }\ {\rm that}\ e_i+e_j\in \c.
\end{equation}

Thus, $\c$ contains at least $n/2$ small words in $\Sp (0,2)$.

Now, consider the maximal subset of elements of  $\c \cap \Sp (0,2)$ with disjoint supports. After reordering the indexes, one can assume that these elements are $e_1+e_2, \ldots, e_{k-1}+e_k$, for some $k \leq n$. We will get the result by considering separately the situations ``$k$ is large'' and ``$k$ is small''.

\medbreak

\noindent {\it Step 3.} If $k\geq n/4$, then for all odd $i\leq k$ and  $s\notin \{i, i+1\}$, consider the sphere $\Sp (e_i+e_{i+1}+e_s,1)$. Since $\c$ is orthogonal to any sphere of radius $1$ and contains $e_i+e_{i+1}$, it should contain at least one other element of $\Sp (e_i+e_{i+1}+e_s,1)$. This other element is either in $ \Sp (0, 2)$ or in $\Sp (0, 4)$. 
\begin{itemize}
\item If this other element of $\Sp (e_i+e_{i+1}+e_s,1) \cap \c$ is in $\Sp (0, 2)$, then it is either $e_i+e_s$ or $e_{i+1}+e_s$. This additional element is in at most one other sphere of the form $\Sp (e_j+e_{j+1}+e_t, 1)$ with $j$ odd, $j\leq k$ and $t\notin \{j, j+1\}$.
\item If this other element is in $\Sp (0, 4)$, then it is of the form $e_i+e_{i+1}+e_s+e_t$ for some $t \notin \{i, i+1, s\}$. For obvious degree reasons this additional element is in at most $4$ spheres of degree $1$ centred at a small word of weight $3$.
\end{itemize}

\begin{figure}[!h]
  \centering
  \includegraphics[scale=0.4]{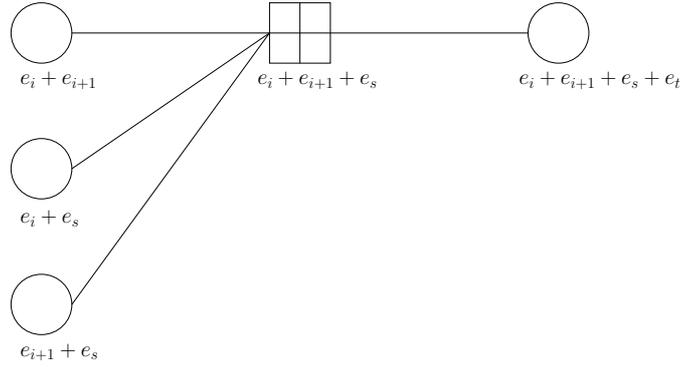}
  \caption{The nodes involved in Step 3.}
\end{figure}

\noindent Finally, Since there are $\frac{k}{2}(n-2)$ spheres of the form $\Sp (e_j+e_{j+1}+e_t, 1)$ with $j$ odd, $j\leq k$ and $t\notin \{j, j+1\}$, there are at least $\frac{1}{4} \frac{k}{2}(n-2)$ additional elements in $\c$ lying in $\Sp (0, 2)\cup \Sp (0,4)$. Therefore, considering also $0$ and the elements $e_1+e_2, \ldots , e_{k-1}+e_k$ we get
$$
wt(\c)\geq 1+\frac{k}{2}+\frac{1}{8}k(n-2).
$$
Since, by assumption $k\geq n/4$ we conclude that
\begin{equation}
  \label{eq:klarge}
wt(\c)\geq 1+\frac{n}{8}+\frac{1}{32}n(n-2) \geq \frac{n^2}{32} \cdot  
\end{equation}

\medbreak

\noindent {\it Step 4.} Now, assume that $k \leq n/4$.
From (\ref{eq:exists}) and by maximality of the set $\{e_1+ e_2, \ldots , e_{k-1}+e_k\}$, for all $\ell\geq k$, there exists at least an integer $j\leq k$ such that $e_j+e_{\ell}\in \c$. For all $\ell \geq k$ let us choose such an integer $i_{\ell}$ such that $i_{\ell}\leq k$ and $e_{i_{\ell}}+e_{\ell}\in \c$.

\medbreak

\noindent {\it Claim.} For all $\ell >k$, there exist at least $\frac{n}{2}-1$ integers $s\in \{1, \ldots , n\}$ such that $e_{i_{\ell}}+e_s \notin \c$.

Indeed, if there were $t\geq \frac{n}{2}$ integers $s_1, \ldots, s_t$ such that $e_{i_{\ell}}+e_{i_s}\in \c$, then
$$
wt(\c+\Sp (e_{i_{\ell}}, 1))>wt(\c),
$$
which contradicts (\ref{eq:weight_c}).

\medbreak

Let us consider the spheres $\Sp (e_{i_{\ell}}+e_{\ell}+e_s, 1)$ for $\ell >k$, $s>k$ and $e_s+e_{i_{\ell}}\notin \c$. Thanks to the previous Claim, we know that there exists at least $(n-k)(\frac{n}{2}-1-k)$ such spheres.
By definition $\c$ is orthogonal to any sphere of radius $1$. In particular $\c$ is orthogonal to $\Sp (e_{i_{\ell}}+e_{\ell}+e_s, 1)$ which contains $e_{i_{\ell}}+e_{\ell}$. Consequently, $\c$ contains at least another vertex of $\Sp (e_{i_{\ell}}+e_{\ell}+e_s, 1)$. Since, by assumption $e_{i_{\ell}}+e_s \notin \c$, the additional vertex is of the form
\begin{enumerate}[(a)]
\item\label{item:ls} either $e_{\ell}+e_s$;
\item\label{item:ilst} or $e_{i_{\ell}}+e_{\ell}+e_s+e_t$ for some integer $t\notin \{i_{\ell}, \ell, s, t\}$.
\end{enumerate}

\begin{figure}[!h]
  \centering
  \includegraphics[scale=0.4]{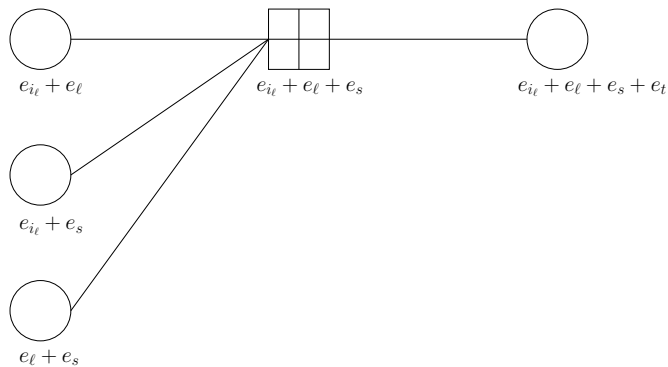}
  \caption{The nodes involved in Step 4.}
\end{figure}

\noindent Case (\ref{item:ls}) cannot happen since it would contradict the maximality of the set $\{e_1+e_2, \ldots , e_{k-1}+e_k\}$.
Since, for obvious degree reasons, a vertex of type (\ref{item:ilst}) is contained in at most $4$ spheres of radius $1$ centred at a small word of weight $3$, then the spheres of the form $\Sp (e_{i_{\ell}}+e_{\ell}+e_s, 1)$ (whose number is at least $(n-k)(\frac{n}{2}-1-k)$) provide at least 
$\frac{1}{4}(n-k)(\frac{n}{2}-1-k)$ additional vertices.
Considering the vertex $0$ and the $\frac{k}{2}$ vertices $e_1+e_2, \ldots , e_{k-1}+e_k$ together with the above-described set of additional vertices, we get
\begin{equation}
  \label{eq:Smallk}
wt(\c)\geq 1+\min_{0\leq k\leq n/4}
\left\{\frac{k}{2}+\frac{1}{4}\left(n-k\right)\left(\frac{n}{2}-1-k\right) \right\}=\frac{3n}{16}\left( \frac{n}{4}-1\right)=\frac{3n^2}{64}-\frac{n}{16}+1 \cdot
\end{equation}

\medbreak

\noindent {\it Final step.} Compare (\ref{eq:klarge}) and (\ref{eq:Smallk}). For all $n$, we have $\frac{n^2}{32}\leq \frac{3n^2}{64}-\frac{n}{16}+1$ and hence $wt(\c)\geq \frac{n^2}{32} \cdot$
\end{proof}

\begin{rem}
  In \cite{cdz11}, the statement \cite[Lemma 8]{cdz11} corresponding to Lemma \ref{lem:card_ball} of the present article is false since it refers to big words of $\C (\F_2^n, S_n)$ and not $\C (\F_2^{m}, S_m \cup W)$.
But the result does not hold for $\C (\F_2^n, S_n)$.
Indeed, spheres of radius $1$ are elements of $\C (\F_2^n, S_n)$ which cannot have a weight quadratic in $n$.
It is then necessary to state the result for elements of $\C (\F_2^m, S_m \cup W)$ even if in the proof we reason locally and can do as if we worked in $\G (\F_2^n, S_n)$.
\end{rem}

\begin{proof}[Proof of Theorem \ref{thm:min_dist}]
  Let $\c$ be a minimum weight codeword in $\C (S_m \cup W)^{\bot} \setminus \C(S_m \cup W)$. From Corollary \ref{cor:vertex_trans}, one can assume that $0\in \c$.
Set $t:= \lfloor \frac{d-1}{2} \rfloor $.
 From Proposition \ref{prop:unbound}, this word is not contained in $\B (0, t-2)$. Thus, in the worst case $\c \subset \B (0, t-1)$.


\medbreak

\noindent {\it Claim.} There are no two consecutive integers $i, i+1\leq t -2$ such that $\c \cap \Sp (0, i)= \c \cap \Sp (0, i+1)=\emptyset$.

Indeed, if both sets were empty, then $\c$ would split into two disjoint sets $\c_1 \cup \c_2$ where $\c_1 = \c \cap \B (0, i-1)$ and $\c_2:=\c \setminus \c_1$. Since the distance between $\c_1$ and $\c_2$ is at least $2$, any sphere of radius $1$ intersects at most one of the words $\c_1, \c_2$. Hence, since $\c$ is orthogonal to any sphere of radius $1$, so are $\c_1$ and $\c_2$. Thus, $\c_1, \c_2 \in \C(S_m\cup W)^{\bot}$ and, by definition of $\c$, at least one of them is not in $ \C(S_m\cup W)$. This contradicts the assumption ``$\c$ has minimum weight in $\C (S_m \cup W)^{\bot} \setminus \C(S_m \cup W)$''.

\medbreak

Thanks to the Claim and Proposition \ref{prop:unbound}, one shows that there exist at least $\lceil \frac{t-1}{10} \rceil$ disjoint balls of radius $4$ centred at an element of $\c$.
Indeed, the worst case is sketched as follows: $\B (0, 4)$ covers every element of $\c$ of weight $\leq 4$, then, from the claim, there exists at least one element $m_1\in \c$ in $\Sp (0, 9)\cup \Sp (0, 10)$. The worst case is when $m_1$ has weight $10$. Then consider the ball $\B (m , 4)$, which is clearly disjoint from $\B (0, 4)$. By the same manner, one uses the claim to assert the existence of an element $m_2$ of $\c$ of weight $19$ or $20$ and consider $\B (m_2, 4)$ and so on...

One shows easily that $\lceil \frac{t-1}{10} \rceil \geq \frac{d}{20}$.
Then, using Lemma \ref{lem:card_ball}, we get the result.
\end{proof}

\section{The Quantum code associated to the classical repetition code}\label{sec:repet}

In this section we answer  a question raised by Mitchison {\it et al.} in \cite{MMS}. Namely, we give the exact parameters of the Quantum code $Q_n$ associated to the classical pure repetition code. That is the Quantum code associated to the Cayley graph $\G (\F_2^n, S_n \cup \{(11\ldots 1)\})$, where $n$ denotes an odd integer and $S_n$ denotes the canonical basis of $\F_2^n$.

In what follows, $n\geq 3$ is an odd integer, $S_n'$ denotes the set of generators $S_n':=S_n \cup \{(11\ldots 1)\}$
and $H_n$ is the $(n-1)\times n$ parity--check matrix of the $[n,1,n]$ 
repetition code whose columns consist of the elements of $S_n'$.
$$H_n=
 \left(
 \begin{array}{cccc}
 1 &        & 0  & 1\\ 
   & \ddots &   & \vdots\\
 0  &        & 1 & 1\\
 \end{array}
 \right).
$$

Our goal is to prove the following theorem.
\begin{thm}\label{thm:main}
The Quantum code associated to the repetition code, i.e. to $\G (\F_2^n, S_n \cup \{(11\ldots 1)\})$ has parameters:
$$
[[N = 2^{n}, K = 2^{\frac{n+1}{2}}, D = 2^{\frac{n-1}{2}}]].
$$
\end{thm}

The matrix  $\G (\F_2^n, S_n \cup \{(11\ldots 1)\})$ has row weight $n+1$ that is logarithmic in the length of the quantum code. This proves the LDPC character of the quantum code.

Actually, the parameters of this family of Quantum codes can be
slightly improved since, from Remark \ref{rem:bipart}, the Cayley graph $\G (\F_2^n, S_n')$ is bipartite. Considering only the vertices corresponding to small words of even weight, one obtains another Quantum code whose length is divided by $2$ and which has the same rate and the same minimum distance. That is, we get a Quantum code with parameters
$$
[[N = 2^{n-1}, K = 2^{\frac{n-1}{2}}, D = 2^{\frac{n-1}{2}}]].
$$

\subsection{Big codes and matrices}

In the previous sections, most of the proofs were combinatorial and involved set of vertices of the Cayley graph. For such a task the terminology of big codes as sets of sets of vertices of the Cayley graph was adapted.

In what follows, we will reason on the matrices $\M (S)$. For this reason we will use preferentially the terminology of kernels and images of the matrix instead those of the big codes $\C (S)$. Notice that, if $\# S$ is even, then $\M (S)$ is self--orthogonal. It is also always symmetric and hence
\begin{eqnarray}
  \label{eq:codes_kernels}
  \Ker \M (S) & = & \C (S)^{\bot}\\
  \im \M (S) & = & \C (S).
\end{eqnarray}

\subsection{The  matrices $\M (\F_2^{n}, S_{n}')$}
Sorting the vectors of $\F_2^{3}$ in the lexicographic order, one obtains

\begin{equation}\label{eq:MS3}
\M (S_3') = 
\left(
\begin{array}{cccccccc}
0 & 1 & 1 & 0 & 1 & 0 & 0 & 1\\
1 & 0 & 0 & 1 & 0 & 1 & 1 & 0\\
1 & 0 & 0 & 1 & 0 & 1 & 1 & 0\\
0 & 1 & 1 & 0 & 1 & 0 & 0 & 1\\
1 & 0 & 0 & 1 & 0 & 1 & 1 & 0\\
0 & 1 & 1 & 0 & 1 & 0 & 0 & 1\\
0 & 1 & 1 & 0 & 1 & 0 & 0 & 1\\
1 & 0 & 0 & 1 & 0 & 1 & 1 & 0
\end{array}
\right). 
\end{equation}

\begin{rem}\label{rem:rankMS3}
This matrix has rank 2. Therefore, the associated classical code has dimension 6 and the Quantum code encodes 4 qubits.  
\end{rem}

\begin{nota}
For every positive integer $s$, let $J_{s}$ be the $s \times s$ matrix defined by
$$
J_{s}:=\left(
  \begin{array}{ccccc}
     &  &  &  & 1 \\
     & (0) &  & 1 &  \\
     & & \iddots &  & \\
     & 1 & &(0) & \\
    1 &  &  &  &  
  \end{array}
\right).
$$
To avoid heavy notation and for convenience sake, we frequently remove the index, which can be easily guessed thanks to the context, and just write $J$.
For the same reason, we frequently write $I$ for the identity matrix, without any index.
\end{nota}
The matrices $\M (S_n')$ can be built recursively in the following manner. 

\begin{lem}\label{lem:recursion_matrix}
Let $n\geq 3$ be an integer, sorting the elements of $\F_2^n$ and $\F_2^{n+1}$ in the lexicographic order, we get
\begin{equation}\label{eq:recursion}
\M (S_{n+1}') =
\left(
\begin{array}{cc}
\M(S_n')+J_{2^n} & I_{2^n}+J_{2^n}\\
I_{2^n}+J_{2^n} & \M (S_n') + J_{2^n}
\end{array}
\right).
\end{equation}
\end{lem}

\begin{proof}
From {\it Remark} \ref{rem:oriented_symmetric}, the matrix is symmetric and hence it is sufficient to prove the result for the upper half of the rows, that is the first $2^n$ rows.
These rows correspond to vectors of $\F_2^{n+1}$ whose $(n+1)$--th coordinate is zero.
Using this lexicographic order, elements of $\F_2^{n+1}$ are in one--to--one correspondence with the integers $0, 1, \ldots , 2^{n+1}-1$, each integer corresponding to the word yielding its dyadic expansion.

\medbreak

\noindent {\it Step 0.} Before reasoning recursively, notice that for all $n$, all the anti-diagonal entries of $\M (S_n')$ are equal to $1$. Indeed, $S_n'$ contains the word $(1\ldots 1)$ and for all $x\in \F_2^{n+1}$, the word $x+(1\ldots 1)$ is the word obtained from $x$ by swapping $0$'s and $1$'s. In term of the above-described correspondence, if $(x)$ is the dyadic expansion of $p\in \{0, \ldots , 2^{n+1}-1\}$, then $x+(1\ldots 1)$ corresponds to $2^{n+1}-1-p$. This yields the terms $J$ in the top right-hand and bottom left-hand blocks of (\ref{eq:recursion}).

\medbreak

\noindent {\it Step 1.} For all $x\in \F_2^{n+1}$ whose $(n+1)$--th entry is $0$.
The corresponding integer $p$ is in $\{0, \ldots , 2^{n}-1\}$ and $x+e_{n+1}$ corresponds to $p+2^n$. This yields the term $I$ in the top right-hand block. 

\medbreak

\noindent {\it Step 2.} The top left-hand block of $\M (S_{n+1}')$ is similar to the matrix $\M (S_n')$ with only one difference, the contribution of $e_1+\cdots +e_n$ in $S_n'$ should be removed from the block of $\M (S_{n+1}')$. This explains the term $J$ in the top left-hand corner.
\end{proof}

\noindent {\bf Caution.} The matrix $\M (S_n')$ is self--orthogonal only for $\# S_n'$ even, that is for $n$ odd. However, because of this recursive approach, it is necessary to consider also the matrices $\M (S_n')$ with even $n$ which do not provide a Quantum code. This is the reason why Lemma \ref{lem:recursion_matrix} is stated for any integer $n\geq 3$ and not only for odd such integers.

\subsection{Computation of the dimension}

\begin{lem}\label{lem:Jsquare}
For all integer $s$, we have $J_s^2 = I_s$.
\end{lem}

\noindent Using the symmetries of the matrix $\M (S_n')$ we obtain the following Lemmas.

\begin{lem}\label{lem:conjugation}
For all odd integer $n\geq 3$, we have:
\begin{enumerate}[(i)]
\item\label{item:conj} $J\M (S_n')J=\M (S_n')$;
\item\label{item:ker} $\c \in \Ker \M (S_n') \Leftrightarrow J \c \in \Ker \M (S_n')$;
\item\label{item:im} $\c \in \im \M (S_n') \Leftrightarrow J \c \in \im \M (S_n')$;
\end{enumerate}
\end{lem}

\begin{proof}
From Lemma \ref{lem:Jsquare}, the left-hand term of (\ref{item:conj}) is the conjugation of $\M (S_n')$ by $J$. Conjugating a matrix by $J$ is nothing but changing the basis by reversing the sorting of its elements, i.e. reversing the sorting of the rows and the columns. In terms of small words it corresponds to apply the permutation of $\F_2^n$ given by the affine automorphism $\phi : x\longmapsto x+ (11\ldots 1)$, which is a Hamming--isometry.
From Corollary \ref{cor:vertex_trans}, the permutation $\phi$ is a graph automorphism of $\G (S_m')$
Thus considering the elements $w_1, \ldots , w_{2^n}$ of $\F_2^n$ sorted by the lexicographic order or sorted as $\phi (w_1), \ldots, \phi (w_{2^n})$ provides the same adjacency matrix.
Another way to prove the assertion is to look at (\ref{eq:MS3}) and observe that it is true for $n=3$. Then to prove the result by induction on $n$ using Lemma \ref{lem:recursion_matrix}.

Assertions (\ref{item:ker}) and (\ref{item:im}) are straightforward consequences of (\ref{item:conj}).
\end{proof}

\begin{lem}\label{lem:caracterisation}
Let $\c = (\c_1, \c_2, \c_3, \c_4) \in \F_2^{2^{n+2}}$ where $\c_i$ are vectors of $\F_2^{2^{n}}$.
Then, we have $\c \in \Ker \M(S_{n+2}')$ if and only if:
\begin{equation}\label{eq:characterisation}
\left\{
\begin{array}{llll}
\c_4 = \c_1 + \d_1 \text{ where } \d_1 \in \Ker \M (S_{n}')\\
\c_3 = \c_2 + \d_2 \text{ where } \d_2 \in \Ker \M (S_{n}')\\
\M (S_n')\c_1 = \d_2 + J\d_1\\
\M (S_n')\c_2 = \d_1 + J\d_2
\end{array}
\right.
.
\end{equation}
\end{lem}

\begin{proof}
By the recursion formula of Lemma \ref{lem:recursion_matrix}, we have:
\begin{equation}\label{eq:2recurs}
\M (S_{n+2}') =
\left(
\begin{array}{cc|cc}
\M (S_{n}')+J & I & I & J\\
I & \M (S_{n}')+J & J & I\\
 & & & \\
\hline
 & & & \\
I & J & \M (S_{n}')+J & I\\
J & I & I & \M (S_{n}')+J
\end{array}
\right).
\end{equation}

This gives a characterisation of the vectors of the kernel $\M (S_{n+2}')$ in function of $\M (S_{n}')$. We have $(\c_1, \c_2, \c_3, \c_4) \in \Ker \M (S_{n+2}')$ if and only if
\begin{eqnarray*}
&
\Leftrightarrow
&
\left\{
\begin{array}{llll}
\M (S_{n}')\c_1 = (\c_2+\c_3) + J(\c_1+\c_4)\\
\M (S_{n}')\c_2 = (\c_1+\c_4) + J(\c_2+\c_3)\\
\M (S_{n}')\c_3 = \M (S_{n}')\c_2\\
\M (S_{n}')\c_4 = \M (S_{n}')\c_1
\end{array}
\right.\\
&
\Leftrightarrow
&
\left\{
\begin{array}{llll}
\c_4 = \c_1 + \d_1 \text{ where } \d_1 \in \Ker \M (S_{n}')\\
\c_3 = \c_2 + \d_2 \text{ where } \d_2 \in \Ker \M (S_{n}')\\
\M (S_{n}')\c_1 = \d_2 + J\d_1\\
\M (S_{n}')\c_2 = \d_1 + J\d_2
\end{array}
\right.
\end{eqnarray*}
\end{proof}

\begin{prop}
For $n$ odd, we have:
$\dim \Ker \M (S_n') = 2^{n-1} + 2^{\frac{n-1}{2}}.$
\end{prop}

\begin{proof}
We prove the result by induction on $n\geq 3$ odd. From Remark \ref{rem:rankMS3}, the case of $\M (S_3')$ is done. Assume now that the result holds for some $n\geq 3$.

If $\c \in \Ker \M (S_{n+2}')$, then the characterisation of Lemma \ref{lem:caracterisation} provides $\d_1, \d_2 \in \Ker \M (S_n')$ such that $\d_1+J\d_2$ and $\d_2+J\d_1$ are in the image of the matrix $\M (S_n')$.
We will show that given such a pair $\d_1, \d_2$ together with a couple of elements of $\Ker \M (S_n')$ one can construct any element of $\Ker (S_{n+2}')$.
First, to study these couples $(\d_1, \d_2)$, let us introduce the map
$$
\varphi :
\left\{
  \begin{array}{ccc}
     \Ker \M (S_n') \times \Ker \M (S_n') & \longrightarrow & \Ker \M (S_n') / \im \M (S_n')\\
(\d_1, \d_2) & \longmapsto & \d_1+J\d_2
  \end{array}
\right.
$$
From Lemma \ref{lem:conjugation}, $\d_1+J\d_2$ and $\d_2+J\d_1$ are both in $\im \M (S_n')$ if and only if $(\d_1, \d_2)$ is in the kernel of $\varphi$.

Given such a couple $(\d_1, \d_2)$, we can construct a codeword in $\Ker \M (S_{n+2}')$ by choosing arbitrary pre-images of $\d_1+J\d_2$ and $\d_2+J\d_1$ for $\c_1$ and $\c_2$. From this, one can construct $\Ker \M (S_{n+2}')$ from $\Ker \varphi$ and $\Ker \M (S_n')$.
Let us choose an arbitrary linear section $L$ of the map $\F_2^{2^n} \rightarrow \im \M (S_n')$ defined by the matrix  $\M (S_n')$. That is, $L$ is a linear map $L: \im \M (S_n') \rightarrow \F_2^{2^{n}}$ such that $\M (S_n')(L(\a)) = \a$ for all $\a \in \F_2^{2^n}$.
Let us introduce the map
$$
\Psi: \left\{
  \begin{array}{ccc}
    \Ker \varphi \times (\Ker \M (S_n'))^2 & \longrightarrow & \Ker \M (S_{n+2}')\\
 & & \\
(\d_1, \d_2, \s_1, \s_2)
& \longmapsto & 
\left(
\begin{array}{llll}
\c_1 = L(\d_2+J\d_1)+\s_1\\
\c_2 = L(\d_1+J\d_2)+\s_2\\
\c_3 = \c_2+\d_2\\
\c_4 = \c_1+\d_1
\end{array}
\right)
  \end{array}
\right.
$$
This map is injective since $\Psi(\d, \s) = \Psi(\d', \s')$ implies $\c_1+\c_4 = \c_1'+\c_4'$ and $\c_2+\c_3=\c_2'+\c_3'$, which entails $\d_1 = \d_1'$ and $\d_2 = \d_2'$. Then $(\c_1, \c_2)=(\c_1',\c_2')$ yields $(\s_1, \s_2) = (\s_1', \s_2')$.
Now, let us show that $\Psi$ is surjective. Let $(\c_1, \c_2, \c_3, \c_4)\in \Ker \M (S_{n+2}')$. From the characterisation (\ref{eq:characterisation}), one gets $\d_1$ and $\d_2$ and $\M (S_n')\c_1 =  \d_2+J\d_1$ means that $\c_1$ is congruent to $L(\d_2 + J\d_1)$ modulo $\Ker \M (S_n')$, which yields $\s_1$. One gets $\s_2$ by the very same manner.

We proved that $\Psi$ is an isomorphism.  Thus,
$$ \dim \Ker \M (S_{n+2}') = 2\dim \Ker \M (S_n') +\dim \Ker \varphi.$$
By the rank-nullity theorem, we get $\dim \Ker \varphi = 2^{n}$, since $\varphi$ is surjective. Finally we have:
$$
\dim \Ker \M (S_{n+2}') = 2\dim \Ker \M (S_n') + 2^{n}= 2^{n+1}+2^{\frac{n+1}{2}}.
$$
\end{proof}

We know that that the number of encoded qubits is $N-2\Rk \M (S_n')$. From the above proposition, we deduce the dimension of the Quantum code.

\subsection{Computation of the distance}
To compute the minimum distance of the Quantum code, we examine the weight of the vectors of $\Ker \M (S_n') \backslash \Ker \M (S_n')^\perp$. That is the set $\Ker \M (S_n') \backslash \im \M (S_n')$.

\begin{lem}\label{lem:representative}
Every word $\c$ of $\Ker \M (S_{n+2}')/\im \M (S_{n+2}')$ satisfies one of the following assertions.
\begin{enumerate}[(i)]
\item $\c$ is of the form $(\c_1, \c_2, \c_2+\d_2, \c_1+\d_1)$ with $\d_1, \d_2 \notin \im \M (S_n')$;
\item  $\c$ has a representative modulo $\im \M (S_n')$ of the form $(\c_1, 0, 0, \c_1)$ with $\c_1 \in \Ker \M (S_n')$.
\end{enumerate}
\end{lem}

\begin{proof}
Let $\c$ be a vector of $\Ker \M (S_{n+2}')$. 
From Lemma \ref{lem:caracterisation}, $\c$ is of the form $(\c_1, \c_2, \c_2+\d_2, \c_1+\d_1)$ where $\d_1, \d_2 \in \Ker \M (S_n')$. Thus, there only remains to prove the statement when either $\d_1$ or $\d_2$ are in $\im \M (S_n')$.

Actually, if one of them is an element of $\im \M (S_n')$, then so is the other one. Indeed, assume that $\d_1 \in \im \M (S_n')$, then using $\M (S_n')\c_1=\d_2+J\d_1$ and Lemma \ref{lem:conjugation}(\ref{item:im}), we see that $\d_2 \in \im \M (S_n')$.
Thus, assume that $\d_1, \d_2 \in \im \M (S_n')$ and let $\b_1, \b_2$ be respective premimages of them.
Thanks to the recursive description (\ref{eq:2recurs}), we know that the vectors of $\im \M (S_{n+2}')$ are of the form:
\begin{equation}\label{eq:description_im}
\left(
\begin{array}{llll}
\M (S_n')\a_1 + J(\a_1+\a_4) + (\a_2+\a_3)\\
\M (S_n')\a_2 + J(\a_2+\a_3) + (\a_1+\a_4)\\
\M (S_n')\a_3 + J(\a_2+\a_3) + (\a_1+\a_4)\\
\M (S_n')\a_4 + J(\a_1+\a_4) + (\a_2+\a_3)
\end{array}
\right),
\end{equation}
where $\a_1, \a_2, \a_3, \a_4 \in \F_2^{2^{n}}$.
Therefore, set $\a_1 = \a_2 = 0$, $\a_3 = \b_2$ and $\a_4 = \b_1$. The vector $\c':=(\b_2+ J\b_1, J (\b_2 +J \b_1), J (\b_2+J\b_1)+\d_2, (\b_2+J\b_1)+\d_1) \in \im \M (S_{n+2}')$.
Thus, replacing $\c$ by $\c+\c'$, which does not changes its class in $\Ker \M(S_n')/\im \M (S_n')$, one can assume that $\d_1=\d_2=0$.
Thus, from now on, $\c$ is for the form $(\c_1, \c_2, \c_2, \c_1)$ and $\M (S_n')\c_2 = \d_1 + J\d_2 = 0$. Set $\a_2 = J\c_2$ and $\a_1=\a_3=\a_4=0$, the vector $\c'':=(J\c_2, \c_2, \c_2, J\c_2)$ is in $\im \M (S_{n+2}')$. Thus, replacing $\c$ by $\c+\c''$ we get a representative of the form $\c = (\c_1, 0, 0, \c_1)$ with $\c_1 \in \Ker \M (S_n')$.
\end{proof}

\begin{prop}\label{prop:wrong_proof}
The minimum distance of the Quantum code $Q_n$ is:
$$
D_n = 2^{\frac{n-1}{2}}.
$$
\end{prop}

\begin{proof}
For $n = 3$, using (\ref{eq:MS3}), we can see that the distance of the Quantum code is $2$. Indeed every non zero codeword has weight at least 2 and for instance the word $e_2+e_3 = (01100000)$ is in the kernel of $\M (S_3')$ although it is not a sum of rows.

We show the result by induction on $n$ for $n\geq 3$ odd.
We proceed as follows, first, we show that the distance of $\Ker(\M (S_{n+2}'))\setminus \im \M (S_{n+2}')$ is bounded below by $2^{\frac{n+1}{2}}$. For that, we consider separately the two situations described by  Lemma \ref{lem:representative}. Then, we show that our lower bound for the minimum distance is reached.

\medbreak

\noindent {\it Step 1.} Let $\c \in \Ker \M (S_{n+2}')$.
Assume that we are in the first case of Lemma \ref{lem:representative}.
That is
$\c$ is of the form $(\c_1, \c_2, \c_2+\d_2, \c_1+\d_1)$, where $\d_1$ and $\d_2$ are in $\Ker \M (S_n')\setminus \im (\M (S_n'))$. By induction and by definition of the Quantum distance, we have $wt(\d_i) \geq 2^{\frac{n-1}{2}}$. Using the triangle inequality for the Hamming distance, we get:
\begin{eqnarray*}
wt(\c_1) + wt(\c_1+\d_1) & = & d(0, \c_1) + d(\c_1, \d_1)\\
& \geq & d(0, \d_1)\\
& \geq & 2^{\frac{n-1}{2}}.
\end{eqnarray*}
Applying the same reasoning to $\d_2$, we show that the weight of $\c$ is at least $2^{\frac{n+1}{2}}$.

\medbreak

\noindent {\it Step 2.}
Now, assume that we are in the second case of Lemma \ref{lem:representative}, that is $\c$ has a representative of the form $\c = (\c_1, 0, 0, \c_1)$ with $\c_1 \in \Ker \M (S_n')$.
We first show that the result holds for the representative and then show the general case.
Assume that $\c_1 \in \im \M (S_n')$. Let $\b_1$ be a pre-image of $\c_1$. Set $\a_1 = \a_4 =\b_1$ and $\a_2 = \a_3 = 0$, using (\ref{eq:description_im}), we see that $\c$ is in the image of $\M (S_{n+2}')$ and hence its weight is not involved in the computation of the minimum distance.
Otherwise $\c_1 \in \Ker \M (S_n') \backslash \im \M (S_n')$ and, by definition of the distance, the weight of $\c$ is at least twice the minimum distance of $\Ker \M (S_n')\setminus \im \M (S_n')$, that is, by induction hypothesis, $2^{\frac{n+1}{2}}$.

\medbreak

Now, let us show that the bound holds while adding an element of $\im \M (S_{n+2}')$ to $\c$.
Let $\x$ be an element of $\im \M (S_{n+2}')$. From (\ref{eq:description_im}), it is of the form
$$
\x =
\left(
\begin{array}{llll}
\M (S_n')\a_1 + \b\\
\M (S_n')\a_2 + J\b\\
\M (S_n')\a_3 + J\b\\
\M (S_n')\a_4 + \b
\end{array}
\right)
$$
for $\b := \a_2+\a_3 +J(\a_1+\a_4)$.
We look at the weight of the two first components of $\c+\x$:
\begin{equation}
  \label{eq:truc}
wt(\c_1+\M (S_n')\a_1+\b) + wt(\M (S_n')\a_2 + J\b),
\end{equation}
the two other components can be treated by the very same manner.
Notice that $J$ is a permutation matrix and hence a Hamming automorphism of $\F_2^{2^n}$, then, using Lemmas \ref{lem:Jsquare} and \ref{lem:conjugation}(\ref{item:conj})
\begin{equation}\label{eq:bidule}
wt(\M (S_n')\a_2+J\b) = wt(J \M (S_n')\a_2 +J^2\b)=wt( \M (S_n')J\a_2 + \b).
\end{equation}

\noindent Combining (\ref{eq:truc}) and (\ref{eq:bidule}) and using the triangle inequality, we get

\begin{eqnarray*}
 wt(\c_1+\M (S_n')\a_1+\b) + wt(\M (S_n')\a_2 + J\b) 
& =    & d(\c_1+\M (S_n')\a_1 ,\b) + d(\M (S_n')J\a_2 , \b) \\
& \geq & d(\c_1 + \M(S_n')\a_1 , \M (S_n')J\a_2)\\
& \geq & wt(\c_1 + \M(S_n')(\a_1+J\a_2))\\
& \geq & d(\c_1, \im \M(S_n')).
\end{eqnarray*}
Since $\c_1 \notin \im \M(S_n')$, then, by induction hypothesis,
we have $ d(\c_1, \im \M(S_n')) \geq 2^{\frac{n-1}{2}}$. Thus,
$wt (\c+\x) \geq 2^{\frac{n+1}{2}}$.

\medbreak

\noindent {\it Final Step.}
We now have a lower bound for the minimum distance.
 Actually, the distance of $\Ker \M(S_{n+2}')\setminus \im \M (S_{n+2}')$ is exactly $2^{\frac{n+1}{2}}$. Indeed, let $\d_1, \d_2$ be two minimum weight words in $\Ker \M (S_n')\setminus \im \M (S_n')$, by induction hypothesis, their weights are $2^{\frac{n-1}{2}}$. The vector $(0, 0, \d_2, \d_1)$ is in $\Ker \M (S_{n+2}')$ and its weight is exactly $2^{\frac{n+1}{2}}$. This vector is not in $\im \M (S_n')$ otherwise, using (\ref{eq:description_im})  we would get $\d_1 = \M (S_n')(\a_1+\a_4)$, which yields a contradiction.
\end{proof}

\section{Concluding Remarks}

\begin{itemize}
\item We have proved a lower bound on the minimum distance of the quantum
  code associated with a classical code with MacKay et al.'s
  construction. This bound is in $O(dn^2)$ where $n$ is the length and
  $d$ is the minimum distance of the classical code. This result is
  based on the enumeration of the minimum number of vertices of a big
  codeword of $\C (S_m \cup W)^{\bot}\setminus \C(S_m \cup W)$
  restricted to a ball of radius 4. We have found it difficult to
  extend this enumeration process to larger balls. We conjecture
  however that the minimum distance of the quantum code is in fact exponential in $d$.

\item This family of quantum codes shares some characteristics with
  topological codes \cite{Kit03a, BM06, BM07a}. The minimum distance
  of a stabilizer code defined on a square lattice 
in two dimensions is subjected to the upper bound $D \leq
\sqrt{N}$. MacKay et al's construction can be seen as a topological
code defined on a lattice with growing dimension. Such a stabilizer
code is not {\em a priori} limited by the Bravyi and Terhal bound \cite{BT07}.

\item 
Some quantum LDPC codes lend themselves to fault-tolerant quantum
computing \cite{RH07, RHG07} and this provides extra motivation for
their study. It would be worthwhile to investigate the potential of
the LDPC family investigated in this paper for such a purpose, and it
would therefore be desirable to understand what kind of logical
operations can be implemented on the encoded data without decoding. A
starting point for this research could be based on Lemma \ref{lem:representative} that leads to a representation of the encoded qubits for the quantum code associated with the repetition code.

\end{itemize}

\section*{Acknowledgment}
This work was supported by the French ANR Defis program under contract
ANR-08-EMER-003 (COCQ project). We acknowledge support from the D\'el\'egation G\'en\'erale pour l'Armement (DGA) and from the Centre National de la Recherche Scientifique (CNRS).

\end{document}